\documentclass[12pt,a4paper]{article}
\usepackage[plain]{fullpage}

\usepackage{paralist}
\usepackage{latexsym}
\usepackage{amsmath}
\usepackage{amsthm}
\usepackage{amssymb}
\usepackage{color}

\newtheorem{theorem}{Theorem}
\newtheorem{lemma}[theorem]{Lemma}
\newtheorem{proposition}[theorem]{Proposition}
\newtheorem{corollary}[theorem]{Corollary}
\newtheorem{fact}[theorem]{Fact}
\newtheorem{remark}[theorem]{Remark}
\newtheorem{claim}[theorem]{Claim}

\def\N{{\mathbb N}}

\def\C{{\mathbb C}}
\def\F{{\mathbb F}}

\newcommand{\ket}[1]{{|{#1}\rangle}}
\newcommand{\bra}[1]{{\langle{#1}|}}

\newcommand{\size}[1]{| #1 |}

\newcommand{\poly}{\mathrm{poly}}

\newcommand\Mat{{\mathcal{M}}}
\newcommand\GL{{\mathrm{GL}}}

\newcommand\AGL{{\mathrm{AGL}}}

\newcommand\Tr{{\mathrm{Tr}}}
\newcommand\tr{{\mathrm{tr}}}

\newcommand{\cH}{{\cal H}}
\newcommand{\cL}{{\cal L}}

\newcommand{\zerovec}{0}
\newcommand{\hsp}{\mathrm{HSP}}
\newcommand{\rhsp}{\mathrm{rHSP}}

\begin{document}

\title{An efficient quantum algorithm for
finding hidden parabolic subgroups in the general linear 
group\footnote{A preliminary version if this paper
appeared in \cite{dikqs-mfcs}.}.
}


\author{
Thomas Decker
\thanks{Centre for Quantum Technologies,
 National University of Singapore, Singapore 117543 {\tt t.d3ck3r@gmail.com}.}
\and
G\'abor Ivanyos
\thanks{Institute for Computer Science and Control, Hungarian 
Academy of Sciences, 
Budapest, Hungary 
({\tt Gabor.Ivanyos@sztaki.mta.hu}).} 
\and 
Raghav Kulkarni 
\thanks{Centre for Quantum Technologies,
 National University of Singapore, Singapore 117543 ({\tt 
 kulraghav@gmail.com}).}
\and
 Youming Qiao 
 \thanks{Centre for Quantum Computation and Intelligent Systems, 
 University of Technology, Sydney; and Centre for Quantum Technologies,
 National University of Singapore, Singapore 117543 
 ({\tt jimmyqiao86@gmail.com}).}
\and
 Miklos Santha 
\thanks{LIAFA, Univ. Paris 7, CNRS, 75205 Paris, France;  and 
Centre for Quantum Technologies, National University of Singapore, 
Singapore 117543 ({\tt miklos.santha@liafa.jussieu.fr}).}
}


\maketitle

\begin{abstract}
In the theory of algebraic groups, parabolic subgroups form a crucial building 
block in the structural studies. In the case of general linear groups over a 
finite field $\F_q$, given a 
sequence of positive integers $n_1, \dots, n_k$, where $n=n_1+\dots+n_k$, a 
parabolic subgroup of parameter $(n_1, \dots, n_k)$ in  
$\GL_n(\F_q)$ is a conjugate of the subgroup consisting of block lower 
triangular matrices where the $i$th block is of size $n_i$. Our main result is 
a quantum algorithm of time polynomial in $\log q$ and $n$ for
solving the hidden subgroup problem in 
$\GL_n(\F_q)$, when the hidden subgroup is promised to be a parabolic subgroup. 
Our algorithm works with no prior knowledge of the parameter of the hidden 
parabolic subgroup. Prior to this work, such an efficient quantum algorithm was 
only known for the case $n=2$ 
(A.~Denney, C.~Moore, and A.~Russell (2010),
Quantum Inf. Comput., Vol. 10, pp. 282-291)
and
for minimal parabolic subgroups (Borel subgroups), 
for the case when $q$ is not much smaller than $n$ (G. Ivanyos: 
Quantum Inf. Comput., Vol. 12, pp. 661-669).
\end{abstract}


\section{Introduction}
\label{sec:intro}

\subsection{Background}\label{subsec:background}

The hidden subgroup problem (HSP for short) is 
defined as follows. A 
function $f$ on a group $G$ is said to hide a subgroup 
$H\leq G$, if $f$ satisfies the following: $f(x)=f(y)$ if and 
only if $x$ and $y$ are in
the same left coset of $H$ (that is, $x^{-1}y\in { H}$). When such an 
$f$ is given as a black box, the HSP asks to determine the hidden subgroup 
$H$.  
Note that the problem when the
level sets of the hiding $f$ are demanded to be right cosets of $ H$ 
-- that is, $f(x)=f(y)$ if and only if $yx^{-1}\in { H}$ -- is equivalent:
composing $f$ with taking inverses maps a  hiding function via right cosets
to a hiding function via left cosets, and vice versa. When we explicitly want to refer
to this variant of the problem, we speak about HSP via right cosets.

The complexity of a hidden subgroup algorithm is measured
in terms of the number of bits representing the elements
of the group $G$, which is usually $O(\log |G|)$. On classical computers, 
the problem has exponential query complexity even for abelian groups. In 
contrast, the quantum query complexity of HSP for any group is polynomial 
\cite{ehk04}, and the HSP for abelian groups can be solved in polynomial time 
with a quantum computer \cite{bl95,kit95}. The latter algorithms are 
generalizations of Shor's result on order finding and computing discrete 
logarithms \cite{Shor}. These algorithms can be further generalized to compute 
the structure of finite commutative black-box groups \cite{cm01}. 

To go beyond the abelian groups is well-motivated by its connection with the 
graph isomorphism problem. Despite considerable attention, the groups for which 
the HSP is tractable remain close to being abelian. For example, we know 
polynomial-time algorithms for the following cases: groups whose derived 
subgroups are of 
constant derived length and constant exponent \cite{FIMSS-stoc}, Heisenberg 
groups \cite{bcvd05,bacon08} and
more generally two-step 
nilpotent groups \cite{iss08},
 ``almost Hamiltonian'' groups 
\cite{gavinsky04}, and groups with a large abelian subgroup and reducible to 
the abelian case \cite{ilg07}. The limited success in going beyond the abelian 
case indicates that the nonabelian HSP may be hard, and \cite{Regev} shows some 
evidence for this by providing a connection between the HSP in dihedral groups 
and some supposedly difficult lattice problem. 

Instead of considering various ambient groups, another direction is to pose 
restrictions on the possible hidden subgroups. This can result in efficient 
algorithms, even over fairly nonabelian ambient groups. For example, if the 
hidden subgroup is assumed to be normal, then HSP can be solved in quantum 
polynomial time in groups for which there are efficient quantum Fourier transforms 
\cite{gsvv01,hrt03}, and even in a large class of groups,
including solvable groups \cite{ims03}. The methods 
of \cite{mrrs04,gpc09} are able to find sufficiently large non-normal 
hidden subgroups in certain semidirect products efficiently. 

Some restricted subgroups of the general linear groups were also considered in 
this context. The result by Denney, Moore and Russell in \cite{dmr10} is 
an efficient quantum algorithm that solves the HSP in the group of 2 
by 2 invertible matrices (and related groups) where the hidden subgroup is 
promised to be a so-called Borel subgroup. 
In \cite{iva12}, Ivanyos considered finding Borel subgroups in general linear 
groups of 
higher degree, and presented an efficient algorithm when the size of the 
underlying field is not much smaller than the degree. 



A well-known superclass of the family of Borel subgroups is the family of 
parabolic subgroups, whose definition is given below.
In this work, we follow the line of research in 
\cite{dmr10,iva12}, and consider the problem of 
finding parabolic subgroups in general linear groups. Our main result will be a 
polynomial-time quantum algorithm for this case, without restrictions on field 
size. 


\subsection{Parabolic subgroups of the general linear 
group}\label{subsec:def}

Let $q$ be a power of a prime $p$. The field with $q$ elements is 
denoted by 
$\F_q$. The vector space $\F_q^n$ consists of  
column vectors of length $n$ over $\F_q$. 
$\GL_n(\F_q)$ stands for the general linear group 
of degree $n$ over $\F_q$. 
The elements of
$\GL_n(\F_q)$ are the invertible $n\times n$ 
matrices with entries from $\F_q$. 
We also use $\GL(V)$ to denote the group of linear automorphisms of the
$\F_q$-space $V$. With this notation, 
we have $\GL_n(\F_q)\cong \GL(\F_q^n)$ and throughout
the paper we will identify these two groups. 
As a matrix
is represented by an array of $n^2$ elements from
$\F_q$, an algorithm is considered efficient if
its complexity is polynomial in $n$ and $\log q$.

We now present the definition of parabolic subgroups (see~\cite{Springer}). 
For a positive integer $k$, and a sequence of positive integers $n_1, \dots, 
n_k$ with $n_1+\dots +n_k=n$, 
the {\em standard parabolic subgroup} of $\GL_n(\F_q)$ with parameter $(n_1, 
\dots, n_k)$ is the subgroup 
consisting of the invertible lower
block triangular matrices of diagonal block sizes
$n_1,\ldots,n_k$. Any conjugate of the standard parabolic subgroup is called a 
{\em parabolic subgroup}. 

To see the geometric meaning of parabolic subgroups, we review the concept of 
flags of vector spaces. Let $\zerovec$ also denote the zero vector 
space. For $\F_q^n$ and $k\geq 1$, a flag ${ F}$ with the parameter 
$(n_1, \dots, 
n_k)$ is a nested sequence of subspaces of $\F_q^n$, that is $\F_q^n=U_0 > U_1 
> U_2 > \dots > U_{k-1}> U_k = \zerovec$, such that for $0\leq i\leq k-1$, 
$\dim(U_i)=n_{i+1}+\dots + n_k$.  $k$ is called the length of $F$.
For $g\in \GL_n(\F_q)$, $g$ stabilizes the flag $F$ if for every $i\in[k]$, 
$g(U_i)=U_i$. 
Then all group elements in 
$\GL_n(\F_q)$ stabilizing ${ F}$ form a parabolic subgroup. On the other hand, 
any parabolic subgroup corresponds to some flag $F$, namely it consists of the 
elements in $\GL_n(\F_q)$ stabilizing $F$. 

For example, the standard parabolic subgroup $B$ in $\GL_5(\F_q)$ 
with 
parameter $(2, 2, 1)$ consists of invertible matrices of the form 
{\tiny
$
\begin{pmatrix}
* & * & 0 & 0 & 0  \\
* & * & 0 & 0 & 0 \\
* & * & * & * & 0 \\
* & * & * & * & 0 \\
* & * & * & * & * 
\end{pmatrix}
$}.
Let $\{e_1, \dots, e_5\}$ be the standard basis of $\F_q^5$. The flag 
stabilized by $B$ is $\F_q^5 > \langle e_3, e_4, e_5\rangle > \langle 
e_5\rangle > \zerovec$. 

A parabolic subgroup is maximal if there are no parabolic subgroups properly 
containing it. It is minimal if it does not properly contain any parabolic 
subgroup. A parabolic subgroup $B$ in $\GL_n(\F_q)$ is 
maximal if and only if it is the stabilizer of a flag of length $2$, that is, 
it is the 
stabilizer of some nontrivial subspace. On the other hand, $B$ is minimal 
if it stabilizes a flag of length $n$. {\em Borel subgroups} in $\GL_n(\F_q)$ 
are just minimal parabolic subgroups.
They are conjugates of the subgroup of invertible lower triangular
matrices.




\subsection{Our results}\label{subsec:results}

The main result of this paper is a polynomial-time 
quantum algorithm for 
finding parabolic subgroups in general linear groups. 

\begin{theorem}\label{thm:main}
Any hidden parabolic subgroup in $\GL_n(\F_q)$ can be found in quantum 
polynomial time (i.e., in time $\poly(\log q,n)$).
\end{theorem}

Note that this algorithm does not require one to know the parameter of the 
hidden
parabolic subgroup in advance. Neither does it pose any restriction on the 
underlying field 
size, while the algorithm in \cite{iva12} for finding Borel subgroups 
requires the field size to be large enough. 
The basic idea behind the algorithm is that in certain cases
the superposition of the elements in a coset of the subgroup
is close to a superposition of the elements of a linear space
of matrices. The latter perspective allows the use of standard algorithms for 
abelian HSPs. Another crucial idea is to make use of the subgroup of common 
stabilizers of all the vectors on a random hyperplane, and examine its 
intersection with the hidden parabolic subgroup. 

We can also consider certain subgroups of 
Borel subgroups, namely the {\em full 
unipotent subgroups}. They are conjugates of the subgroup of lower triangular 
matrices with $1$'s on the diagonal. Following a variant of the idea for 
Theorem~\ref{thm:main}, we can show that these subgroups can be efficiently
found if the base field is small.

\begin{theorem}\label{thm:main2}
Any hidden full unipotent subgroup in $\GL_n(\F_q)$ can be found by a
quantum algorithm in time $\poly(q,n)$.
\end{theorem}


Finally, 
we consider finding the maximal parabolic subgroups in the classical setting. 
We show that in the classical setting,  the 
deterministic and randomized query complexities are exponential, in contrast to 
the efficient quantum algorithm as above.

\begin{theorem}\label{thm:main3} 
For $d\leq n/2$ the query complexity for a bounded-error randomized algorithm 
with bounded error probability 
$\epsilon$ 
to find a maximal parabolic subgroup stabilizing a $d$-dimensional 
subspace in $\GL_n(\F_q)$ is $\Omega(q^{d/2})$. 
\end{theorem}

The proof is based on the fact that for any $o(q^{d/2})$ matrices
which are not scalar multiples of each other there are still many 
$d$-dimensional subspaces such that the matrices fall into pairwise
distinct cosets of the corresponding maximal parabolic subgroups.
As every Borel subgroup is contained in a unique maximal parabolic
subgroup stabilizing an $n/2$-dimensional subspace, the same
argument gives the following.

\begin{corollary}\label{cor:main4}
The query complexity for a bounded-error randomized algorithm 
with bounded error probability 
$\epsilon$ 
to find a hidden Borel subgroup in $\GL_n(\F_q)$ is $\Omega(q^{\lfloor n/4\rfloor})$. 
\end{corollary}

\noindent{\it The structure of the paper.} In Section~\ref{sec:prel} we collect 
certain preliminaries for the paper. In particular, in Section~\ref{sec:QFT}
 we adapt the 
standard algorithm for abelian HSP to linear subspaces, which forms the basis 
of our algorithms. We then present an efficient quantum algorithm for finding  
maximal parabolic subgroups in Section~\ref{subsec:stab_space}. 
Section~\ref{sec:tool} describes a main technical tool, a 
generalization of the result of \cite{mrrs04,dmr10} for finding complements
in affine groups. In Section~\ref{sec:algo} we present the algorithm for 
finding parabolic subgroups, proving Theorem~\ref{thm:main}. 
In Section~\ref{subsec:unipotent} we 
consider the task of finding unipotent subgroups, proving Theorem~\ref{thm:main2}. 
In Section~\ref{sec:classical} we discuss
the deterministic and randomized complexities of finding hidden maximal Borel 
subgroups 
in the classical setting, proving Theorem~\ref{thm:main3}
and Corollary~\ref{cor:main4}. Finally in Section~\ref{sec:conclude} we conclude this paper and propose some future directions. 

\section{Preliminaries
}\label{sec:prel}

\subsection{Notations and facts}\label{subsec:fact}

Throughout the article, $q$ is a prime power. For 
$n\in \N$, $[n]=\{1, \dots, n\}$. $\Mat_n(\F_q)$ is the set of $n\times n$ 
matrices over $\F_q$. For a finite group $G$, we will be concerned with 
finding a subgroup $H$ in $G$, when it is promised that $H$ is from a fixed 
family of subgroups $\cH$. We use $\hsp(G, \cH)$ to denote the 
HSP problem with this promise, and $\rhsp(G, \cH)$ to denote the HSP via right 
cosets of $H\in \cH$. 
Let $V$ be a vector space. For a subspace $U\leq V$ and $G=\GL(V)$, let $G_U$ 
be the subgroup in $G$ consisting of elements that act as
\emph{pointwise stabilizers} on $U$. That is, $G_U=\{X\in \GL(V) : \forall u\in 
U, Xu=u\}$. 
Let $G_{\{U\}}$ be the subgroup in $G$ 
consisting of elements that act as \emph{setwise 
stabilizers} on $U$. That is, $G_{\{U\}}=\{X\in\GL(V) : XU=U\}$. Note that 
$\{G_{\{U\}} : 0<U<V\}$ is just the set of maximal parabolic subgroups. 

\begin{fact}\label{fact:frac_nonsing}
For every prime power $q$, and for every positive integers $n\geq m$, the 
probability for a random 
$n\times m$ matrix $M$ over $\F_q$ to have rank $m$ is
no less than what we have in the case of $q=2$, that is $\frac{1}{2}\cdot 
\frac{3}{4}\cdot \frac{7}{8}\cdot \dots \approx 
0.288788 > 1/4$. 
\end{fact}

\subsection{The quantum Fourier transform of linear spaces}
\label{sec:QFT}

In this part we briefly discuss slight generalizations
of the Fourier transform of linear spaces over $\F_q$ 
introduced in \cite{iva12} and a version  
useful for certain linear spaces of matrices.
Let $V\cong \F_q^m$ be a linear space over the field $\F_q$
and assume that we are given a nonsingular
symmetric bilinear function $\phi:V\times V\rightarrow \F_q$.
By $\C^V$ we denote the Hilbert space of dimension
$q^m$ having a designated orthonormal basis consisting
of the vectors $\ket{v}$ indexed by the elements $v\in \F_q^m$.

Let $q=p^r$ where $p$ is a prime and let
$\omega$ be the primitive $p$th root 
$e^{\frac{2\pi i}{p}}$ of unity. 
We define the quantum Fourier transform 
with respect to $\phi$
as the linear transformation $QFT_\phi$ of $\C^V$ which maps
$$\ket{v}\mbox{~~to~~}
\frac{1}{\sqrt{\size V}}\sum_{u\in V}\omega^{\Tr(\phi(u,v))}\ket{u},$$
where $v\in V$ and
$\Tr$ is the trace map from $\F_q$ to $\F_p$
defined as $\Tr(x)=\sum_{i=0}^{r-1}x^{p^i}$.
It turns out that $QFT_\phi$ is a unitary map
and, if the vectors from $V$ are represented by
arrays of elements from $\F_q$ that are
coordinates in terms of an orthonormal basis
of $V$ with respect to $\phi$ 
(that is, $\phi$ is the standard inner product of
$\F_q^m$)
then 
$QFT_\phi$ is just the $m$th tensor power
of the QFT defined in \cite{vdhi06} for $\F_q$. 
(This is the linear transformation
of $\C^{\F_q}$ that maps $\ket{x}$ ($x\in \F_q$) to 
$\frac{1}{\sqrt q}\sum_{y\in \F_q}\omega^{\Tr(xy)}\ket{y}$.)
Therefore, in this case, by Lemma~2.2 of \cite{vdhi06}, 
$QFT_\phi$ has a polynomial time approximate implementation 
on a quantum computer. In the general case, where elements of
$V$ are represented by coordinates in terms of a not necessarily
orthonormal basis w.r.t. $\phi$, the map $QFT_\phi$ can
be efficiently implemented by composing the above transform
with linear transformations of $\C^V$ corresponding to 
appropriate basis changes for $V$.

For a subset $A\subseteq V$ we adopt the standard notation 
$\ket{A}$ for the
uniform superposition of the elements of $A$, that is
$\ket{A}=\frac{1}{\sqrt{|A|}}\sum_{a\in A}\ket{a}.$
Assume that we receive the uniform superposition
$
\ket{v_0+W}=\frac{1}{\sqrt{\size {W}}}\sum_{v\in {W}}\ket{v_0+v}
$
over the a coset $v_0+W$ of the
$\F_q$-linear subspace 
$W$ of $V$ and for some $v_0\in V$.
Let $W^{\perp}$ stand for the subspace of $V$
consisting of the vectors $u$ from $\F_q^m$ such that
$\phi(u,v)=0$ for every $v\in W$. 
By results from \cite{iva12}, 
if we measure the state after the Fourier transform,
we obtain a uniformly random element 
of ${W}^\perp$.
If instead of the uniform superposition over the coset
$v_0+W$ we apply the QFT to the superposition 
$\ket{v_0+W'}=\frac{1}{\sqrt{\size{W'}}}\sum_{v\in W'}
\ket{v_0+v}$
over a subset $v_0+W'$ for $\emptyset\neq W'\subseteq W$, the resulting state 
is
$\sum_{u\in V}c_u'\ket{u},$
where 
$$c_u'=\bra{u}QFT_\phi\ket{v_0+W'}=
\frac{\omega^{\Tr\phi(v_0,u)}}{\sqrt{\size{{W'}}\size{V}}}
\sum_{v\in {W'}}\omega^{\Tr\phi(v,u)}.
$$
For $u\in W^\perp$ we have
\begin{equation}\label{eq:qft_subset}
|c_u'|=\frac{\size{{W'}}}{\sqrt{\size{{W'}}\size{V}}}=
\frac{\sqrt{\size{W'}}}{\sqrt{\size{W}}}\cdot
\frac{1}{\sqrt{\size{{W}^\perp}}},
\end{equation}
\noindent whence, after measurement the chance of obtaining
a particular $u\in W^\perp$ is $\frac{\size{W'}}{\size{W}}$ 
times as much as if we had in the case of the uniform distribution over
$W^\perp$.

In this paper we consider subspaces and certain subsets
of the linear space $\Mat_{n}(\F_q)$. If we take the 
inner product $\phi_0(A,B)=\tr(AB^T)$
  the elementary matrices form an orthonormal basis.
It follows that $QFT_{\phi_0}$, being just the 
$n^2$th tensor power of the QFT of $\F_q$, can be efficiently 
approximated. However, for the purposes of this paper it turns out to be
more convenient using the inner product $\phi(AB)=\tr(AB)$.
The map $QFT_{\phi}$ is the composition of $QFT_{\phi_0}$ with
taking transpose (the latter is just a permutation of the matrix entries).
The main advantage of considering $QFT_{\phi}$ is that it is invariant
in the following sense: we always obtain the same $QFT_{\phi}$ even
if we write matrices of linear transformations of 
the space $V=\F_q^n$ in terms of various bases.
In particular, in our hidden subgroup algorithms we can think of
 our matrices in terms of a basis a priori unknown to us
in which the hidden subgroup has a natural form, for example
lower block triangular.

\subsection{A common procedure for HSP algorithms}\label{subsec:common}

Suppose we want to find some hidden subgroup $H$ in $G=\GL_n(\F_q)$. Let 
$V=\F_q^n$. We present 
the standard procedure that produce 
a uniform superposition over a coset of the hidden subgroup. This
part will be common in (most of) the hidden subgroup algorithms
presented in this paper. 
First we show how to produce the uniform superposition over $\GL(V)$.
The uniform superposition 
$\frac{1}{q^{n^2}}\sum_{X\in \Mat_n(\F_q)}\ket{X}$
over $\Mat_n(\F_q)$ can be produced using the QFT
for $\F_q^{n^2}$. Then, in an additional qubit we compute
a Boolean variable according to whether or not
the determinant of $X$ is zero. We measure this qubit,
and abort if it indicates that the matrix $X$ has determinant zero.
This procedure gives the uniform superposition over $\GL(V)$ with
success probability more than $\frac{1}{4}$. 

Next we assume that we have the uniform superposition 
$\frac{1}{\sqrt{\size{\GL(V)}}}\sum_{X}\ket{X}\ket{0}$, summing over $X\in 
\GL(V)$.
Recall that $f$ is the function hiding the subgroup. We appended a new quantum 
register, 
initialized to zero, for holding
the value of $f$. We compute $f(X)$ in this second register, 
measure and discard it. The result is
$\ket{A{ H}}=
\frac{1}{\sqrt{\size{ H}}}\sum_{X\in { H}}\ket{AX}$
for some unknown $A\in\GL(V)$. $A$ is actually uniformly random,
but in this paper we will not make use of this fact.

\section{Maximal parabolic subgroups}\label{subsec:stab_space}
  
In this section, 
we settle the HSP when 
the hidden subgroup is 
a maximal parabolic subgroup, which will be used in the main algorithm 
in Section~\ref{sec:algo}. It also helps to illustrate the idea of 
approximating a subgroup in the general linear group by a 
subspace in the linear space of matrices. 
  
Recall that a parabolic subgroup ${ H}$ is maximal if it stabilizes some 
subspace $0<U<\F_q^n$.
We mentioned in 
Section~\ref{subsec:fact} that they are 
just setwise stabilizers of subspaces. Determining ${ H}$ is equivalent to 
finding $U$. Set $V=\F_q^n$.

\begin{proposition}\label{prop:max}
Let $G=\GL_n(\F_q)$, and $\cH=\{G_{\{U\}} : 0<U<V\}$. $\hsp(G, 
\cH)$ can be 
solved in quantum polynomial time. 
\end{proposition}
\begin{proof}
Let $H$ be the hidden maximal parabolic subgroup, stabilizing some 
$(n-d)$-dimensional subspace $U\leq \F^n$. Note that $d$ is unknown to us. 
Before describing the algorithm, we observe the following: checking correctness 
of a guess for $U$, and hence for
$ H$, can be done by applying the oracle to a set of generators of the
stabilizer of $U$, as there are no inclusions between maximal parabolic
subgroups. 

Now we present the algorithm. First produce a coset superposition $\ket{A{ H}}$ 
for unknown $A\in\GL(V)$, as described in Section~\ref{subsec:common}.
Let 
${ W}=\{X\in \Mat_n(\F_q):XU\leq U\}.$
In a 
basis whose last $n-d$ elements are from $U$,
 ${ W}$ is the subspace of the matrices of the form
$
\begin{pmatrix}
B & \\
C & D 
\end{pmatrix},
$
where $B$ and $D$
are not necessarily invertible, and the empty space in the 
upper right corner means
a $d\times (n-d)$ block of zeros. Noting that such a matrix is invertible if 
and only if $B$ and $D$ are both invertible, we have ${ H}\subset { W}$ and
$\frac{\size{A{ H}}}{\size{A{ W}}}=
\frac{\size{{ H}}}{\size{{ W}}}
>\frac{1}{4\times 4}.$
Also, viewing in 
the same basis, 
$(A{ W})^\perp A=W^\perp$
 consists of the matrices of the
form
$
\begin{pmatrix}
 & \mbox{~} \\
* & 
\end{pmatrix},
$
where $*$ stands for an arbitrary $(n-d)$ times $d$ matrix. This implies that
$(A{ W})^\perp=\{X\in\Mat_n(\F_q): XV\leq U\mbox{~and~}XU=\zerovec\}A^{-1}.$

If $d\geq n/2$, we apply QFT to the \emph{left} coset superposition 
$\ket{A H}$
and perform a measurement, for any element $X$ in $(A{ W})^\perp$, the 
measurement will produce $X$
with probability no less than $\frac{1}{16 \size{(A{ W})^\perp}}$.
It follows that $XA$ will be a particular matrix
 from $(A{ W})^\perp A$ with probability at least 
$\frac{1}{16 \size{(A{ W})^\perp}}$. 
Then more than $\frac{1}{4}$ of the $(n-d)\times d$ matrices
have rank $n-d$. It follows that with probability
at least $\frac{1}{64}$, the matrix $XA$ will be
a matrix from $(A{ W})^\perp A$ whose image is $U$.
As $XV=XAV$, we can conclude that
$XV=U$ with probability more than $\frac{1}{64}$. 

For the case $d < n/2$ we consider the
HSP via \emph{right} cosets of $H$, and let act matrices on row vectors from 
the right. Via the same procedure as above, it will reveal the dual subspace 
stabilized by $H$, which determines $H$ uniquely as well. 

Finally, though $d$ is not known to us, depending on whether $d\geq n/2$, one 
of these two procedures 
with produce $U$ correctly with high probability. So we perform the two 
procedures 
alternatively, and use the checking procedure to determine which produces the 
correct result. This concludes the algorithm. 
\end{proof}
\section{A tool: finding complements in small stabilizers}\label{sec:tool}
In this section, we introduce and partially
settle a new instance of the hidden subgroup 
problem. This will be an important technical tool for the main algorithm.

Consider the hidden subgroup problem in the following setting. The ambient 
group $G\leq \GL_n(\F_q)$ consists of the
invertible matrices of the form
$
\begin{pmatrix}
b & \\
v & I
\end{pmatrix},
$
where $b\in \F_q$, 
$v$ is a column vector from $\F_q^{n-1}$, and $I$ is the 
$(n-1)\times (n-1)$ identity matrix. 
The family of hidden subgroups $\cH$ consists of all conjugates of $H_0$, where 
$H_0$ is the subgroup of diagonal matrices in $G$:
$
{ H}_0=\left\{\begin{pmatrix}
b & \\
 & I
\end{pmatrix}: b\in\F_q^* \right\}.
$
Note that any conjugate of $H_0$ is 
$
{ H}_v=
\left\{
\begin{pmatrix}
b & \\
(b-1)v & I
\end{pmatrix}: b\in\F_q^* \right\},
$
for some $v\in \F_q^{n-1}$. We will consider the HSP via right cosets in this 
setting.

The group $ G$ has an abelian normal
subgroup $ N$ consisting of the
matrices of the form
$
\begin{pmatrix}
1 & \\
v & I
\end{pmatrix}
$
isomorphic to $\F_q^{n-1}$, and the subgroups ${ H}_v$
are the semidirect complements
of $ N$. For $n=2$, $ G$ is the affine group $\AGL_1(\F_q)$. 
The HSP in  $\AGL_1(\F_q)$ is solved in quantum polynomial
time in \cite{mrrs04} over prime fields
and in \cite{dmr10} in the general case
using the non-commutative Fourier transform of the
group $\AGL_1(\F_q)$. The algorithm served as
the main technical ingredient in \cite{dmr10}
for finding Borel subgroups in $\GL(\F_q^2)$.
A generalization for certain similar
semidirect product groups is given in \cite{bcvd05}.
To our knowledge, the first occurrence of the idea of 
comparing with a coset state in a related abelian group is in \cite{bcvd05}. 
Here, due to the ``nice'' representation of the group
elements, we can apply the same idea in a simpler
way, while in \cite{bcvd05} 
it was needed to be combined with a discrete logarithm
algorithm which is not necessary here.

\begin{proposition}\label{prop:tool}
Let $G$ and $\cH$ be as above, and suppose $q=\Omega(n/\log n)$. Then $\rhsp(G, 
\cH)$ can be solved in quantum polynomial time.
\end{proposition}
\begin{proof}
Assume that the hidden subgroup is $H= H_v$ for some $v\in \F_q^{n-1}$.
As right cosets of $H$ are being considered, we have superpositions over right 
cosets 
${ H}A$  for some unknown $A\in { G}$. 
The actual information of each matrix $X$ from $ G$
is contained in $X-I$, a matrix from the $n$-dimensional space
$ L$ of matrices whose last $n-1$ columns are zero.
We will work in $ L$.
Set $$\widetilde { W}'=
\{X-I:X\in{ H}\}=
\left\{
\begin{pmatrix}
b & ~\\
bv & ~
\end{pmatrix}: -1\neq b\in \F_q
\right\}
\text{ and }
{ W}=
\left\{
\begin{pmatrix}
b & ~\\
bv & ~
\end{pmatrix}: b\in \F_q
\right\}.
$$
Then $ W$ is a one-dimensional subspace of $ L$.
It turns  out that ${ W}={ W}A$ for every matrix $A\in { G}$
(that is why it is convenient to consider the HSP via right cosets).
It follows that
$\{(Y+I)A-I:Y\in { W}\}=\{YA+(A-I):Y\in { W}\}={ W}+A-I,$
whence the set $\{XA-I:X\in { H}\} $ equals
${ W}'+A-I$ for 
${ W}' = \widetilde { W}'A$. 

Therefore, after an application of the QFT of $ L$ to the state
$\ket{{ H}A-I}=\ket{{ W}'+A-I}$ and a measurement, we obtain every
specific element of ${ W}^\perp$ with
probability at least $\frac{q-1}{q}\frac{1}{\size{{ W}^\perp}}$.
More generally, if we do the procedure for a product of $n-1$ 
superpositions over right cosets of ${ H}$ we obtain each 
specific $(n-1)$-tuple of vectors from ${ W}^\perp$ 
with probability at least
$(\frac{q-1}{q})^{n-1}\frac{1}{{\size{{ W}^\perp}}^{n-1}}.$
Since the probability that $n-1$ random elements from a space of dimension
$n-1$ over $\F_q$ span the space is at 
least $\frac{1}{4}$, therefore, the probability of getting a basis of ${ 
W}^\perp$ is $\Omega((\frac{q-1}{q})^{n-1})$. Using this basis,
we obtain a guess for $ W$ and $ H$ as ${ H}$ is
the set of invertible matrices from ${ W}+I$.
A correct guess will be obtained expectedly with
$O((\frac{q}{q-1})^{n-1})$ repetitions. This is polynomial
if $q$ is $\Omega(n/\log n)$.
\end{proof}

Finally we note that for constant $q$, or more generally for constant 
characteristic, 
\cite{FIMSS-stoc} can be used to obtain a polynomial time algorithm. On the 
other hand, it is intriguing to study the case of ``intermediate'' 
values of $q$.


\section{The main algorithm}\label{sec:algo}
\subsection{The structure of the algorithm}
In this subsection, we describe the structure of an algorithm for finding 
parabolic subgroups in general linear groups, proving Theorem~\ref{thm:main}. 
Let $G=\GL_n(\F_q)$, $V=\F_q^n$, and the hidden parabolic subgroup $H$ be the 
stabilizer of 
the flag $V>U_1>U_2>\dots > U_{k-1}>0$. Note that the parameter of the flag, 
including $k$, is unknown to us. The algorithm will output the hidden flag, 
from which a generating set of 
the parabolic subgroup can be constructed easily. 

Let $T=U_{k-1}$ denote the smallest subspace in the flag. The 
algorithm relies on the 
following subroutines crucially. These two subroutines are described in 
Section~\ref{subsec:base_case} and 
Section~\ref{subsec:recursion}, respectively. 

\begin{proposition}\label{prop:guess_part}
Let $G$, $H$ and $T$ be as above. There exists a quantum polynomial-time 
algorithm, that given access to an oracle hiding $H$ in $G$, produces three 
subspaces $W_1$, $W_2$ and $W_3$, s.t. one of $W_i$ 
is a nonzero subspace contained in $T$ 
with high probability.
\end{proposition}

\begin{proposition}\label{prop:verify}
Let $G$, $H$ and $T$ be as above. There exists a classical polynomial-time 
algorithm, that given access to an oracle hiding $H$ in $G$, and some $0<W\leq 
V$, determines whether $W\leq T$, and in the case of $W\leq T$, whether $W=T$.
\end{proposition}

Given these two subroutines, the algorithm proceeds as follows. It starts with 
checking whether $k=1$, that is whether $H=G$. This 
can be done easily: produce a set of 
generators of $G$, and check whether the oracle returns the same on all of 
them. If $k=1$, return the trivial flag $V>0$. 

Otherwise, it repeatedly calls the subroutine in 
Proposition~\ref{prop:guess_part} until that subroutine produces subspaces 
$W_1$, $W_2$ and $W_3$, such that for some $i\in[3]$, we have $\zerovec<W_i\leq 
T$. 
This can be verified by
Proposition~\ref{prop:verify}. 
Let $W$ be this subspace. The second subroutine then also tells whether $W=T$. 

After getting $0<W\leq T$, the algorithm fixes a subspace $W'$ to be 
any direct complement of $W$ in $V$, and makes a recursive call to the HSP with 
a new ambient group $G'$, and a new hidden subgroup $H'$, 
as follows. $G'$ is 
$\{X\in \GL(V):XW'\leq W'\mbox{~and~}(X-I)W=\zerovec\},$
which is isomorphic to $\GL(W')\cong \GL(V/W)$. $H'$ is the stabilizer of the  
flag $W'>W'\cap U_1>\dots >W'\cap U_{k-1}\geq 0$. Note that the oracle 
restricted to $G'$ realizes a hiding function for $H'$. 

The recursive call then returns a flag in $W'$ as $W'>U_1'>U_2'>\dots 
>U_{k'}>0$. Let $U_i=\langle U_i'\cup W\rangle$, 
$i\in[k']$. If $W=T$, then the algorithm outputs the flag 
$V>U_1>U_2>\dots > U_{k'} > W>0$. If $W<T$, return $V>U_1>U_2>\dots >U_{k'}>0$. 

It is clear that at most $n$ recursive calls will be made, and the algorithm 
runs in polynomial time given that the two 
subroutines run in polynomial time too. 
We now prove Proposition~\ref{prop:guess_part} and \ref{prop:verify} in the 
next two subsections.



\subsection{Guessing a part of the flag}\label{subsec:base_case}

In this subsection we prove Proposition~\ref{prop:guess_part}. 
Recall that $G=\GL_n(\F_q)$, the hidden subgroup ${ H}$ stabilizing of the flag
$V>U_1>\ldots>U_{k-1}>\zerovec$, and $T=U_{k-1}$. 
The algorithm of \cite{dmr10} for finding hidden Borel
subgroups in 2 by 2 matrix groups was based
on computing the intersection with the stabilizer of
a nonzero vector. Here we follow an extension of the idea 
to arbitrary dimension $n$. We consider the {\em common} stabilizer of $n-1$ 
linearly
independent vectors.

Pick a random subspace $U'\leq V$ of dimension $n-1$. Recall that $G_{U'}$ 
denotes the group of pointwise stabilizers of $U'$.
We also consider the group consisting of
the unipotent elements of $ G_{U'}$,
${ N}=\{X\in \GL(V):(X-I)V\leq U'\mbox{~and~} X\in G_{U'}\}.$
Note that ${ N}$ is an abelian normal subgroup of $ G_{U'}$
of size $q^{n-1}$. Here we illustrate the form of ${ G_{U'}}$ and ${ N}$ 
when $U'$ is put in an appropriate basis:
{\footnotesize
$$
\begin{array}{cc}
\begin{pmatrix}
1 & & & & * \\
& 1 & & & * \\
& & 1 & & * \\
& & & 1 & * \\
 & & & & * 
\end{pmatrix}
, &
\begin{pmatrix}
1 & & & & * \\
& 1 & & & * \\
& & 1 & & * \\
& & & 1 & * \\
 & & & & 1 
\end{pmatrix}.
\\
{ G_{U'}}
&
{ N}
\end{array}
$$}
We will describe three procedures, whose success on producing some $0<W\leq 
T$ depend on $d:=\dim(T)$ 
and the field size $q$. Each of these procedures only works for a certain range 
of $d$ and $q$, but together they cover all possible cases. Thus, the algorithm 
needs to run each of these procedures, and return the three results from them.
The general idea behind these procedures is to examine the intersection of the 
random hyperplane $U'$ 
with $T$. As $d=\dim(T)$, the probability that $U'$ contains $T$
is $\frac{q^{n-d}-1}{q^n-1}\sim \frac{1}{q^d}$. 

Assume first that $U'$ does not contain $T$.
We claim that in this case 
\begin{equation}\label{eq:H_cap_G}
\sum_{X\in { H}\cap { G_{U'}}}(X-I)V=T
\end{equation}
\begin{equation}\label{eq:H_cap_N}
\text{and} \quad \sum_{X\in { H}\cap { N}}(X-I)V=U'\cap T.
\end{equation}
To see this, pick
$v_n\in T\setminus U'$, and let
 $v_1,\ldots,v_{n-1}$ be a basis for $U'$ such that
for every $0<j<k$, the system $v_{n-\dim(U_j)+1},\ldots,v_{n-\dim(U_{j+1})}$
is a basis for $U_j$. In the basis $v_1,\ldots,v_n$, the matrices
of the elements of ${ N}$ are the matrices with ones in the 
diagonal,
arbitrary elements in the last column except the lowest one,
and zero elsewhere. Among these the matrices of the elements of 
intersection with $ H$  are those whose first $n-d$ entries
in the last column are also zero: {\footnotesize
$$
\begin{array}{ccc}
\begin{pmatrix}
* & & & &  \\
* & * & & &  \\
* & * & * & &  \\
* & * & * & * & * \\
* & * & * & * & * 
\end{pmatrix}
, &
\begin{pmatrix}
1 & & & &  \\
& 1 & & &  \\
& & 1 & &  \\
& & & 1 & * \\
 & & & & * 
\end{pmatrix}
, &
\begin{pmatrix}
1 & & & &  \\
& 1 & & &  \\
& & 1 & &  \\
& & & 1 & * \\
 & & & & 1 
\end{pmatrix}.
\\
{ H}
&
{ H} \cap { G_{U'}}
&
{ H} \cap { N}
\end{array}
$$}

Based on the above analysis, the three procedures are as follows.

\begin{compactitem}
\item If $d>1$, then ${ H}\cap { N}$ is nontrivial. As ${ N}$ is 
abelian, we can efficiently compute ${ H}\cap { N}$ by the abelian 
hidden subgroup algorithm. Thus by Equation~\ref{eq:H_cap_N}, we can use it to 
compute $W_1$ as a guess for  a nontrivial subspace of $T$. 
\item If $d=1$ and $q\geq n$, we can
compute ${ H}\cap { G_{U'}}$ in $G_{U'}$ by the algorithm in 
Proposition~\ref{prop:tool},
and use it to compute $W_2$ as a guess for $T$ by 
Equation~\ref{eq:H_cap_G}.

\item If $d=1$ and $q<n$, with probability at 
least
$\frac{1}{q}-\frac{1}{q^2}=\Omega(\frac{1}{q})=\Omega(\frac{1}{n})$,
 we have that $U'\geq T$ but
$U'$ does not contain $U_{k-2}$. Then we have
\begin{equation}\label{eq:last_case}
\sum_{X\in { H}\cap { N}}(X-I)V=U'\cap U_{k-2}.
\end{equation}
To see this, pick $v_n\in U_{k-1}\setminus \{0\}$,
$v_{n-1}\in U_{k-2}\setminus U'$, and $v_1,\ldots,v_{n-2}$
s.t. $v_1,\ldots,v_{n-2},v_n$ is a basis for $U'$ and
for every $0<j<k$, the system $v_{n-\dim(U_j)+1},\ldots,v_{n-\dim(U_{j+1})}$
is a basis for $U_j$. In this basis the matrices for the
elements of ${ N}\cap { H}$ are those whose entries
are zero except the ones in the diagonal and except
the other lowest $\dim U_{k-2}$ entries in the next to
last column:{\footnotesize
$$
\begin{array}{ccc}
\begin{pmatrix}
* & & & &  \\
* & * & & &  \\
* & * & * & * &  \\
* & * & * & * &  \\
* & * & * & * & * 
\end{pmatrix}
, &
\begin{pmatrix}
1 & & & * & \\
& 1 & & * & \\
& & 1 & * & \\
& & & 1 & \\
& & & * & 1 
\end{pmatrix}
, &
\begin{pmatrix}
1 & & & & \\
& 1 & & & \\
& & 1 & * & \\
& & & 1 & \\
& & & * & 1 
\end{pmatrix}.
\\
{ H}
&
{ N} 
&
{ H} \cap { N}
\end{array}
$$}
Again, we can find ${ H}\cap { N}$
by the abelian hidden subgroup algorithm 
and use Equation~\ref{eq:last_case} to compute
$V'=U'\cap U_{k-2}$. If $ \dim V'=1$ then return $W_3=V'$ as the guess for 
$T$. Otherwise we take
a direct complement $V''$ of $V'$ and restrict the HSP
to the subgroup of the transformations $X$ such that $(X-I)V''=0$ and
$XV'\leq V'$ (which is isomorphic to $\GL(V')$) and apply the method in 
Proposition~\ref{prop:max}
to compute a subspace $W_3$ as the guess  for $T$.
\end{compactitem}

\subsection{Checking and recursion}\label{subsec:recursion}
In this subsection we prove Proposition~\ref{prop:verify}. Recall that the goal 
is to determine whether some subspace $0<W\leq V$ is contained in $T=U_{k-1}$,
the last 
member of the flag $V>U_1>\dots>U_{k-1}>0$ stabilized by the hidden parabolic 
subgroup $H$. If $W\leq V$, we'd like to know whether $W=T$. This can be 
achieved with the help of the following lemma.

\begin{lemma}\label{prop:check}
Let ${ H}$ be the stabilizer in $\GL(V)$ of the flag
$V>U_1>U_2>\ldots>U_{k-1}>\zerovec$, and let $0<W<V$. 
Let $W'$ be any direct complement of $W$ in $V$.
Then
$U_{k-1}\geq W$ if and only if 
${ H}\geq \{X\in \GL(V):(X-I)V\leq W\}.$
Furthermore, if $U_{k-1}\geq W$ then
$U_{k-1}=W$ if and only if
$${ H}\cap \{X\in \GL(V):(X-I)V\leq W'\mbox{~and~}(X-I)W'=0\}=\{I\}.$$
\end{lemma}

It is clear that this allows us 
to determine whether 
$U_{k-1}\geq W$: form a generating set of $\{X\in \GL(V):(X-I)V\leq W\}$, 
and query the oracle to see whether all element in the generating set 
evaluate the same. 
Also, 
if $U_{k-1}\geq W$, we can test whether 
$U_{k-1}=W$ by solving an instance of the abelian HSP.

Let us present an intuitive interpretation of this lemma. Consider  
a basis of 
$V$ consisting of a basis of $W'$, 
followed by a basis of $W$. Then the subgroup mentioned in the first part is 
the group of invertible matrices of the form $Y+I$, where the 
first $d=\dim W'$ rows of $Y$ are zero. 
The group of the second part 
consists of the matrices of the form $I+Y$ where
only the upper right $d\times (n-d)$ block of $Y$ can contain nonzero entries.
This is an abelian group.

\begin{proof}
Let $\cL=\{X\in \GL(V):(X-I)V\leq W\}$.
To see that $U_{k-1}\geq W \Rightarrow \cH\geq \cL$, we show that every 
$X\in\cL$ stabilizes the flag. For $i\in\{1, \dots, k-1\}$, and $v\in 
U_i$, 
$(X-I)v\in W\leq U_{k-1}\leq U_i$. Thus $Xv\in U_i$, and $X$ stabilizes 
the 
flag. We prove the other direction $\cH\geq \cL \Rightarrow 
U_{k-1}\geq W$ by contradiction. That is, if $U_{k-1}\not\geq W$, then we 
exhibit
some $X\in \cL\setminus \cH$. For this, choose some nonzero $b\in U_{k-1}$ and 
$c'\in 
U\setminus U_{k-1}$, and form $c=b+c'\not\in U_{k-1}$. Fix a basis of 
$V$ as 
$\{b, c, d_1, \dots\}$. Now consider the linear map $X$ s.t. $X$ switches $b$ 
and $c$, 
and leaves $d_i$'s fixed. It is straightforward to verify that 
$X\in\cL$ and 
$X\not\in \cH$.


For the furthermore part, 
we set $\cL'=\{X\in \GL(V):(X-I)V\leq 
W'\mbox{~and~}(X-I)W'=0\}$. 
To see the if direction, assume that $U_{k-1}>W$. 
Then $U_{k-1}\cap W'\neq \zerovec$ and 
$$\{X\in\GL(V):(X-I)W'=\zerovec\mbox{~and~}(X-I)W\leq U_{k-1}\cap W'\}$$
is a nontrivial subgroup of 
$\cH\cap \cL'$.
For the only if direction, assume that $W=U_{k-1}$ and that
$X\in \cH\cap \cL'$. 
For any $v\in V$, $Xv-v\in W'$ by $X\in\cL'$. We show that $Xv-v\in W$ as well. 
For this, write
$v=w+w'$ where $w\in W=U_{k-1}$ and $w'\in W'$, thus 
$Xv-v=X(w+w')-(w+w')=Xw-w\in 
U_{k-1}=W$ by $X\in \cH\cap \cL'$. This shows that for any $v\in V$, $Xv-v\in 
W\cap W'=\zerovec$, so $X=I$.
%
%
%
%
\end{proof}
\begin{remark}\label{remark:dual}
Instead of considering whether a subspace $W$ is contained in $U_{k-1}$, we can 
also decide whether $W$ contains $U_1$ as follows. Let us consider the same 
hypotheses as 
in Lemma~\ref{prop:check}. Then $U_{1}\leq W$ if 
and only if 
${ H}\geq \{X\in \GL(V):(X-I)W=\zerovec\}.$
Furthermore, if $U_{1}\leq W$ then
$U_{1}=W$ if and only if
$${ H}\cap \{X\in \GL(V):(X-I)V\leq W'\mbox{~and~}(X-I)W'=0\}=\{I\}.$$
With the help of the above claim, there is another 
possible recursion scheme: if 
$V>W\geq U_1$ is found then take any direct complement $W'$ 
of $W$ in $V$ and recurse with the smaller ambient group 
$\{X\in \GL(V):XW\leq W\mbox{~and~}(X-I)W'=\zerovec\}$, which is
isomorphic to $\GL(W)$.
\end{remark}

\section{Finding hidden full unipotent groups}\label{subsec:unipotent}

A full unipotent group in $\GL(V)$ is the subgroup
$${ H}=\{X\in\GL(V):(X-I)U_j\leq U_{j+1}\;(j=0,\ldots,n-1)\}$$
for some complete flag $V=U_0>U_1>\ldots>U_{n-1}>U_n=\zerovec$
of subspaces. The full unipotent groups are the $p$-Sylow 
subgroups of $\GL(V)$ 
(recall that $q$ is a power of the prime $p$)
and finding generators for one of them
is equivalent to finding the corresponding flag. We can use 
a variant of the
method described in Section~\ref{subsec:base_case} to find $U_{n-1}$ in time
$(q+n)^{O(1)}$. 

We pick a random subspace $W'$ of dimension $n-1$.
Put 
$${ N}=\{X\in \GL(V):(X-I)V\leq W'\mbox{~and~}(X-I)W'=0\}.$$
With probability $\Omega(\frac{1}{q})$, we have $W'\cap U_{n-2}=U_{n-1}$. 
If this is the case then 
$$\sum_{X\in { N}\cap { H}}(X-I)V=U_{n-1}.$$
To see this, pick $v_n\in U_{n-1}\setminus \{0\}$,
$v_{n-1}\in U_{n-2}\setminus U_{n-1}$,
and $v_j\in W'\cap U_{j-1}\setminus U_{j}$, for $j=1,\ldots,n-2$.
In this basis the matrices for the
elements of ${ N}\cap { H}$ are those whose entries
are zero except the ones in the diagonal and except
the lowest  entry in the next to last column: 
$$
\begin{array}{ccc}
\begin{pmatrix}
1 & & & &  \\
* & 1 & & &  \\
* & * & 1 &  &  \\
* & * & * & 1 &  \\
* & * & * & * & 1 
\end{pmatrix}
, &
\begin{pmatrix}
1 & & & * & \\
& 1 & & * & \\
& & 1 & * & \\
& & & 1 & \\
& & & * & 1 
\end{pmatrix}
, &
\begin{pmatrix}
1 & & & & \\
& 1 & & & \\
& & 1 &  & \\
& & & 1 & \\
& & & * & 1 
\end{pmatrix}.
\\
{ H}
&
{ N} 
&
{ N} \cap { H}
\end{array}
$$
Therefore we can use the abelian hidden subgroup algorithm
for finding ${ H} \cap { N}$ and use it to
compute a guess for $U_{n-1}$. We can test a guess for
$U_{n-1}$ by testing a set of generators of the group
$$\{X\in \GL(V):(X-I)V\leq U_{n-1}\mbox{~and~}(X-I)U_{n-1}=\zerovec\}$$
for membership in $ H$ and the restrict the hiding function
to the subgroup
$$\{X\in \GL(V):(X-I)U_{n-1}=\zerovec\}\cong\GL_{n-1}(\F_q)$$
in order to find the other members of the flag by recursion.
The complexity of the procedure is $(q+n)^{O(1)}$.

\section{Maximal and minimal
parabolic subgroups: classical algorithms}\label{sec:classical}
In this section, we consider the following HSP: the ambient group 
$G=\GL_n(\F_q)$, and for some integer $0<d<n$, the family of hidden subgroups 
is 
$\cH=\{G_{\{U\}} : U\leq 
V, 
\dim(U)=d\}$, that is those subgroups setwise stabilizing $d$-dimensional 
subspaces. We assume that $d$ is given. 
\subsection{A simple deterministic algorithm}

Here is a simple deterministic algorithm for finding $U$: try every 
hyperplane in $V$ until we obtain a hyperplane $W\geq U$. Once such $W$ is 
obtained we recurse as described in Remark~\ref{remark:dual}. It is also 
described in Remark~\ref{remark:dual} how to check with the oracle
whether $W\geq U$. 
The cost is polynomial in the number of hyperplanes $q^n-1$, which is 
sub-exponential in $n^2\log q$ when $n$ is reasonably large.



\subsection{An almost tight lower bound}
We now present a lower bound 
for the query complexity of a randomized algorithm for this HSP. 
First we present the lower bound for the deterministic case 
based on an adversary strategy. Then we argue that a 
minor adaptation of this works
for the randomized case. We will suppose w.l.o.g that $d\leq n/2$, as otherwise 
we can replace with $d$ by $n-d$. 

\subsubsection{Deterministic query complexity}
\begin{proposition}
Let $G=\GL_n(\F_q)$, and $\cH=\{G_{\{U\}} : U\leq V, \dim(U)=d\}$. Any 
deterministic algorithm for $\hsp(G, \cH)$ must make $\Omega(q^{d/2})$ queries.
\end{proposition}
\begin{proof}
Suppose that the deterministic algorithm queries
the oracle for $N$ group elements. The strategy of the adversary 
is simply to return different values (different labels of cosets) for these 
elements until it becomes impossible.
That is, as long as there still exists a $d$-dimensional subspace which is 
{\em not} stabilized
by any of the non-scalar quotients of pairs of the queried matrices.
In other words, if during the execution of algorithm, the queries are $g_1, ... 
, g_N,$ where assuming without loss of generality that all $g_i$'s are 
distinct, the adversary returns labels $1, 2, ... , N$.
Note that the answer to the $t^{th}$ query adds at most $t$ new pairs of 
quotients of $g_i$s.
Hence there are at most $\binom{N}{2}$ such nontrivial quotients $g_i g_j^{-1}.$
If $g$ is one of the quotients, then the algorithm learns that $g$ does not 
stabilize the hidden subspace.
In order to continue the adversary strategy, all we need to make sure is that 
the quotients generated so far
do not stabilize all the hidden subspaces of dimension $d,$ i.e., there are 
still two $d$-dimensional subspaces
which are not stabilized by any of the quotients. In this case the algorithm 
can not answer correctly on any of these two subspaces.
Thus it suffices to upper bound the number of subspaces stabilized
by an individual group element.

For $a, b\in \N$, $b\leq a$,
let ${\binom{a}{b}}_q$ be the Gaussian binomial 
coefficient, which counts 
the number of $b$-dimensional subspaces of $\F_q^a$. If 
$b> a$
then set 
${\binom{a}{b}}_q=0$. The analogue of the Pascal equality for binomial 
coefficients is
${\binom{a}{b}}_q=q^b{\binom{a-1}{b}}_q+{\binom{a-1}{b-1}}_q$. It is also 
easily deduced that 
${\binom{a}{b}}_q=\frac{q^a-1}{q^{a-b}-1}{\binom{a-1}{b}}_q$.
\begin{claim}
For $A\in \GL_n(\F_q)$, if $A\neq \lambda I$, $\lambda \in\F_q^{\times}$, then 
$A$ can stabilize at most ${\binom{n-1}{d}}_q+{\binom{n-1}{d-1}}_q$
$d$-dimensional subspaces.
\end{claim}
\begin{proof}
We prove by induction on $n$. When $n=d$, this can be verified easily. 
Suppose the claim holds for $d\leq n< k$. Then for $n=k$, we distinguish the 
following cases. 

\noindent{\it Case I.} Suppose there does not exist a hyperplane $P\leq V$, 
s.t. $A$ acts as a scalar matrix on $P$. Then by induction hypothesis, for any 
hyperplane $P\leq V$, the restriction of $A$ on $P$ stabilizes at most  
$({\binom{k-2}{d}}_q+{\binom{k-2}{d-1}}_q)$ $d$-dimensional subspaces. Thus the 
number of $d$-dimensional subspaces that $A$ stabilizes is at most 
$\frac{q^k-1}{q^{k-d}-1}\cdot ({\binom{k-2}{d}}_q+{\binom{k-2}{d-1}}_q)
\leq \frac{q^k-q}{q^{k-d}-q}\cdot {\binom{k-2}{d}}_q + 
\frac{q^{k-1}-1}{q^{k-d}-1} \cdot {\binom{k-2}{d-1}}_q
\leq {\binom{k-1}{d}}_q+{\binom{k-1}{d-1}}_q$. 

\noindent{\it Case II.} Suppose $A$ acts on some hyperplane $P\leq V$ as 
$\lambda I$. If $A$ only stabilizes $d$-dimensional subspaces in $P$, $A$ 
stabilizes at most ${\binom{n-1}{d}}_q$ $d$-dimensional subspaces. Otherwise, 
suppose $A$ stabilizes a subspace $U\leq V$, $\dim(U)=d$, and 
$U\not\leq P$. Take $v\in U\setminus P$, and suppose $Av=\mu v+ w$ for $\mu\in 
\F_q^{\times}$, and $w\in P\cap U$. We now consider the following cases.

\noindent{\it Case II (i).} If $\mu\neq \lambda$, let $\gamma=1/(\mu-\lambda)$. 
Then $A(v+\gamma w)=\mu(v+\gamma w)$. Form a basis of $V$ as $(b_1, \dots, 
b_{k-1}, v+\gamma w)$, where $b_1, \dots, b_{k-1}$ is a basis for $P$. Then 
w.r.t. this basis $A$ is $\mathrm{diag}(\lambda, \dots, \lambda, \mu)$. The 
number of subspaces stabilized by $\mathrm{diag}(\lambda, \dots, \lambda, \mu)$ 
is clearly ${\binom{k-1}{d}}_q+{\binom{k-1}{d-1}}_q$. 

\noindent{\it Case II (ii).} If $\mu = \lambda$, first note that $w\neq 0$, 
since $A\neq \lambda I$. Now 
consider a basis of $V$ 
as $(v, w, b_1, \dots, b_{k-2})$, where $(w, b_1, \dots, b_{k-2})$ is a basis 
of $P$. W.r.t this basis,
 $A$ is of the form 
$$
\begin{pmatrix}
\lambda & & & & \\
1 & \lambda & & & \\
& & \lambda &  & \\
& & & \ddots & \\
& & &  & \lambda
\end{pmatrix}.
$$
Then for any $d$-dimensional $U$ s.t. $A(U)=U$ and $U\not\leq P$, it is easy to 
check that $U$ must contain $w$, so the number of such $U$ is at most 
${\binom{k-1}{d-1}}_q$, which the number of $d$-dimensional subspaces 
containing $w$. This concludes this case.
\end{proof}
Therefore the quotients can hit less than 
$N^2({\binom{n-1}{d}}_q+{\binom{n-1}{d-1}}_q)$ subspaces.
As ${\binom{n}{d}}_q=q^d{\binom{n-1}{d}}_q+{\binom{n-1}{d-1}}_q$
and ${\binom{n-1}{d-1}}_q\leq {\binom{n-1}{d}}_q$ for $d\leq n/2$,
these do not give all the $d$-dimensional subspaces unless $N=\Omega(q^{d/2})$
when $d\leq n/2$. In other words, this yields a lower bound of 
$\Omega(q^{d/2})$ for the deterministic query complexity. 
\end{proof}
Note that when $d$ is 
around $\frac{n}{2}$, this lower bound gives $\Omega(q^{n/4})$, as compared 
to the $q^{O(n)}$ upper bound shown in the last subsection. 
\subsubsection{Randomized query complexity}
Recall that any randomized query algorithm $R$ for HSP that has success 
probability say $\epsilon$ 
can be viewed as a probability distribution over several deterministic 
query algorithms, i.e., the deterministic algorithm $D_i$ is used with 
probability $p_i,$ where $i \in \{1, \ldots, k \}$ for some $k.$ 

On every input, i.e., on every $d$ dimensional subspace $V$ the probability 
that the $D_i$ outputs correctly on $V$
is at least $1 - \epsilon.$ By a simple double-counting argument (cf.  Yao's 
min-max principle applied to uniform distribution on inputs), there is a 
$j \in \{1, \ldots, k\}$ such that the deterministic algorithm $D = D_j$ 
outputs correctly on at least $1 - \epsilon$ fraction of the inputs.
Thus in order to prove lower bound for randomized case, 
it suffices to show a lower bound for the deterministic algorithm that computes 
correctly on $1 - \epsilon$ fraction of the inputs.
\begin{claim}\label{claim:queriesneeded} 
Let $G=\GL_n(\F_q)$, $d\leq n/2$, and $\cH=\{G_{\{U\}} : U\leq V, \dim(U)=d\}$. Let $D$ be a 
deterministic algorithm for $\hsp(G,\cH)$ that answers 
correctly on at least $1 - \epsilon$ fraction (for a constant $\epsilon \geq 0$)
of the $d$-dimensional
subspaces. Then $D$ must make at least $\Omega(q^{d/2})$ queries in worst case.
\end{claim} 
\begin{proof}
Suppose $D$ makes $N$ queries and computes correctly on at least $1 - \epsilon$ 
fraction of inputs.
First we note that the adversary strategy (described in previous subsection) 
for the deterministic query complexity is non-adaptive, i.e., 
one may assume that the answers to the queries are fixed by the adversary 
beforehand. So we can assume without loss of generality 
that  the adversary  answers distinct coset-labels
$1, 2, \ldots , N$ as long as it can be consistent with its answers.
We use the same adversary strategy for $D$ that was used for the deterministic 
case.
Our lower bound will only be weaker by a multiplicative  $1 - \epsilon$ factor.

To see this, 
note that after $N$ queries $D$ obtains information about at most $N \choose 2$ 
quotient elements. Moreover, since $D$ makes error on at most
$\epsilon$ fraction of inputs, there are at most  $\epsilon$ fraction of the 
$d$-dimensional subspaces {\em uncovered} by the stabilizers
of these $N \choose 2$ quotient elements.
Hence, from the point of view of the adversary, it can  can continue giving 
different labels as answer to the queries as long 
as there are still $\epsilon$ fraction of the $d$-dimensional subspaces still 
left uncovered. Hence we get essentially the same lower bound
with the multiplicative factor of $1 - \epsilon$ on the query complexity of 
$D,$ and hence on query complexity of any randomized algorithm.
\end{proof}

This proves Therom~\ref{thm:main3}. 
As every Borel subgroup is contained in the stabilizer of a unique
$\lfloor n/2\rfloor$-dimensional subspace, 
Claim~\ref{claim:queriesneeded} remains 
valid with $d=\lfloor n/2\rfloor$ if we replace $\cH$ with the class of Borel subgroups, 
proving Corollary~\ref{cor:main4}. 

\section{Concluding remarks}\label{sec:conclude}

We have shown that hidden parabolic subgroups of the general linear group
over a finite field can be found in quantum polynomial time. Efficient 
procedures for finding parabolic subgroups in related groups, that is in 
the (projective) special linear group can be derived using the
techniques described in \cite{iva12} for Borel subgroups. One possible
direction for further research could be determining the complexity of 
$HSP(\GL_n(F_q),\cH)$ for other well known classes $\cH$ of large 
subgroups of $\GL_n(\F_q)$,
e.g., where $\cH$ consists of (certain subclasses) the classical groups.
For full unipotent groups we gave a method polynomial in $n$ and $q$. 
Even for $n=2$ it would be interesting to know whether there is an
algorithm of complexity subexponential in $\log q$ 
(e.g, $2^{O(\sqrt{\log q})}$). 

The instance of the HSP discussed in Section~\ref{sec:tool} 
is a problem in the flavor of the generalized hidden shift problem
introduced in~\cite{cvd07}. In our view, the existence of a quantum algorithm for this 
problem which also works in polynomial time where the base 
field is neither of constant characteristic nor sufficiently large
is an interesting open question. A positive answer would also simplify 
the main algorithm of the present paper.

\paragraph*{Acknowledgements.}
The research is partially funded by the Singapore Ministry of Education and the 
National Research Foundation, also through the Tier 3 Grant ``Random numbers 
from quantum processes,'' MOE2012-T3-1-009. 
Research partially supported by the European Commission
IST STREP project Quantum Algorithms (QALGO) 600700,
by the French ANR Blanc program under contract ANR-12-BS02-005 (RDAM project), 
and by the Hungarian Scientific Research Fund (OTKA), Grant NK105645.

%
%
%
%


\begin{thebibliography}{000}
\bibitem{bacon08}
D.~Bacon (2008),
{\it 
How a Clebsch-Gordan transform helps to solve the Heisenberg hidden subgroup
problem},
Quantum Inf. Comput., Vol. 8, pp. 438-467.


\bibitem{bcvd05}
D.~Bacon, A.~Childs, and W.~van~Dam (2005),
{\it {From} optimal measurement to efficient quantum algorithms
for the hidden subgroup problem over semidirect product groups},
In {Proc. 46th IEEE FOCS}, pp. 469-478.

\bibitem{ber68}
E.~R.~Berlekamp (1968), 
{\it Algebraic coding theory},
McGraw-Hill, New York.

\bibitem{ber70}
E.~R.~Berlekamp (1970), 
{\it Factoring polynomials over large finite
fields}, 
Math. Comput., Vol. 24, pp. 713-735.

\bibitem{bl95}
D.~Boneh and R.~Lipton (1995),
{\it Quantum cryptanalysis of hidden linear functions},
In: Proc. Crypto'95, 
pp. 427-437. 

\bibitem{cvd07}
A.~M.~Childs and W.~van~Dam,
{\it Quantum algorithm for a generalized hidden shift problem},
in: Proc. SODA 2007,
pp. 1225-1237. 

\bibitem{cz81}
D.~G.~Cantor, H.~Zassenhaus (1981),
{\it A New Algorithm for Factoring Polynomials Over Finite Field,}
Math. Comput. 36, pp. 587-592.

\bibitem{cm01}
K.~Cheung and M.~Mosca (2001),
{\it Decomposing finite abelian groups},
Quantum Inf. Comput., Vol. 1, pp. 26-32.



\bibitem{dmr10} 
A.~Denney, C.~Moore, and A.~Russell (2010),
{\it Finding conjugate stabilizer subgroups in PSL(2;q) and related
problems},
Quantum Inf. Comput., Vol. 10, pp. 282-291.


\bibitem{vdhi06}
W.~van~Dam, S.~Hallgren, and L.~Ip (2006),
{\it Quantum algorithms for some hidden shift problems},
SIAM J. Comput., Vol 36, pp. 763-778.

\bibitem{dikqs-mfcs}
T.~Decker, G.~Ivanyos, R.~Kulkarni, Y.~Qiao, and M.~Santha (2014),
{\it An efficient quantum algorithm for
finding hidden parabolic subgroups in the general linear group},
In.: Proc. MFCS 2014, 
pp. 226-238.


\bibitem{ehk04}
M.~Ettinger, P.~Hoyer, and E.~Knill (2004),
{\it The quantum query complexity of the hidden subgroup problem 
is polynomial},
Inform. Proc. Lett., 91, pp. 43-48.

\bibitem{FIMSS-stoc}
K.~Friedl, G.~Ivanyos, F.~Magniez, M.~Santha, and P.~Sen (2003),
{\it Hidden translation and orbit coset in quantum computing},
In: Proc. 35th STOC, pp. 1-9.

\bibitem{gavinsky04}
D.~Gavinsky (2004),
{\it Quantum solution to the hidden subgroup problem for
poly-near-Hamiltonian groups}, 
Quantum Inf. Comput. Vol. 4, pp. 229-235.

\bibitem{gpc09}
D.~N.~Goncalves, R.~Portugal and C.~M.~M.~Cosme (2009),
{\it Solutions to the hidden subgroup problem on some metacyclic groups},
In: Proc. TQC2009, 
Lect. Notes Comput. Sci., Vol. 5906, Springer-Verlag (Berlin), 
pp. 1-9.


\bibitem{gsvv01}
M.~Grigni, L.~Schulman, M.~Vazirani, and U.~Vazirani (2001),
{\it Quantum mechanical algorithms for the nonabelian
          Hidden Subgroup Problem},
In {Proc. 33rd ACM STOC}, pp. 68-74.



\bibitem{hrt03}
S.~Hallgren, A.~Russell, and A.~{Ta-Shma} (2003),
{\it Normal subgroup reconstruction and quantum computation
          using group representations},
{SIAM J. Comp.}, 32, pp. 916-934.


\bibitem{ilg07}
Y.~Inui and F.~Le~Gall (2007),
{\it Efficient quantum algorithms for the hidden subgroup problem over
semi-direct product groups},
Quantum Inf. Comput., Vol. 7, pp. 559-570.

\bibitem{iva12}
G.~Ivanyos (2012),
{\it Finding hidden Borel subgroups of the general linear group},
{Quantum Inf. Comput., Vol.  12, pp. 0661-0669.}



\bibitem{ims03}
G.~Ivanyos, F.~Magniez, and M.~Santha (2003),
{\it Efficient quantum algorithms for some instances 
          of the non-{A}belian hidden subgroup problem},
Int. J. Found. Comp. Sci., Vol. 15, pp. 723-739.

\bibitem{iss08}
G.~Ivanyos, L.~Sanselme, and M~Santha (2012),
{\it An efficient quantum algorithm for  the hidden subgroup problem
in nil-2 groups},
{Algorithmica 63(1-2): pp. 91-116.}


\bibitem{jozsa01}
R.~Jozsa (2001),
{\it Quantum factoring, discrete logarithms, and the hidden subgroup
problem},
Computing in Science and Engineering, Vol. 3, pp. 34-43.


\bibitem{kit95}
A.~Yu.~Kitaev (1995),
{\it Quantum measurements and the Abelian Stabilizer Problem,}
Technical report arXiv:quant-ph/9511026.




\bibitem{mrrs04}
C. Moore, D. Rockmore, A. Russell, and L. Schulman (2004),
{\it The power of basis selection in Fourier sampling:
Hidden subgroup problems in affine groups},
In {Proc. 15th ACM-SIAM SODA}, pp. 1106-1115.





\bibitem{Regev}
O.~Regev (2004),
{\it Quantum computation and lattice problems},
SIAM J. Comput. 33, pp. 738-760.





\bibitem{Shor}
P.~Shor (1997),
{\it Algorithms for quantum computation: {Discrete} logarithm and
          factoring},
SIAM J. Comput., 26, pp 1484-1509.

\bibitem{Springer}
T.~A.~Springer (1998),
{\it Linear Algebraic groups},
Progress in mathematics, Vol. 9,
2nd ed.,
Birkh\"auser (Boston).

\bibitem{watrous}
J.~Watrous (2001),
{\it Quantum algorithms for solvable groups},
In {Proc. 33rd ACM STOC}, pp. 60-67.



\end{thebibliography}
\end{document}